\documentclass[letterpaper, 10 pt , conference]{ieeeconf}  

\usepackage[T1]{fontenc}

\IEEEoverridecommandlockouts                              

\usepackage{graphics} 
\usepackage{epsfig} 
\usepackage{times} 
\usepackage{amsmath}
\usepackage{amssymb}  
\usepackage{balance}
\usepackage{tabu}
\usepackage{color}
\usepackage{mathptmx} 
\usepackage{comment}
\usepackage[compress]{cite}

\DeclareMathOperator*{\argmin}{arg\,min}
\usepackage{enumerate}
\usepackage{hyperref}

\newtheorem{remark}{Remark}

\newtheorem{thm}{Theorem}
\newtheorem{asum}{Assumption}

\DeclareMathOperator*{\argmax}{arg\,max}
 \newcommand{\norm}[1]{\left\lVert#1\right\rVert}

\ifodd 1

\newcommand{\com}[1]{{\color{red}{Comment: #1}}}
\else
\newcommand{\com}[1]{}
\fi

\title{\LARGE \bf
A Game-Theoretic Approach to a Task Delegation Problem}

\author{Donya G. Dobakhshari, Lav R. Varshney, and Vijay Gupta
\thanks{D.~G. Dobakhshari and V. Gupta are with the Department of Electrical Engineering, University of Notre Dame, IN, USA. Email:{\tt\small (dghavide, vgupta2)@nd.edu.}}
\thanks{L.~R. Varshney is with the Department of Electrical and Computer Engineering, University of Illinois at Urbana-Champaign, IL, USA.  Email: {\tt\small varshney@illinois.edu.} His work was supported in part by NSF grant CCF-1623821. }
\thanks{This work has been submitted to the IEEE for possible publication. Copyright may be transferred without notice, after which this version may no longer be accessible.}}%

\begin{document}
\maketitle
\thispagestyle{empty}
\pagestyle{empty}
\begin{abstract}
We study a setting in which a principal selects an  agent to execute a collection of tasks according to a specified priority sequence. Agents, however, have their own individual priority sequences according to which they wish to execute the tasks. There is information asymmetry since each priority sequence is private knowledge for the individual agent. We design a mechanism for selecting the agent and incentivizing the selected agent to realize a priority sequence for executing the tasks that achieves socially optimal performance. Our proposed mechanism consists of two parts. First, the principal runs an auction to select an agent to allocate tasks to with minimum declared priority sequence  misalignment. Then, the principal rewards the agent according to the realized priority sequence with which the tasks were performed. We show that the proposed mechanism is individually rational and incentive compatible. Further, it is also socially optimal for the case of linear cost of priority sequence modification for the agents.
\end{abstract}

\section{Introduction}
Consider a situation in which a system operator must hire one among several agents to execute some tasks. The operator has a quality of service (QoS) constraint that implies a desired order in which the tasks should be executed. The agents, however, may prioritize tasks execution in a different order depending on their own private preferences and shifting their preferred order of execution may impose a cost on the agents. Such misalignment of the preferred order of execution among the principal and the agents, especially with information asymmetry, creates performance inefficiency from the principal's viewpoint. Minimizing this inefficiency requires the principal to devise an appropriate mechanism to select an agent and incentivize him to shift his preferred priority order for executing the tasks.

Such a formulation is relevant to many situations. For instance, in a cloud computing application, users request a Cloud Computing Service Provider (CCSP) to perform a job. The CCSP then allocates the tasks among the servers. If tasks come in at a high rate and the number of servers is limited, the tasks may form a task queue~\cite{CaoHLZ2013, armbrust, PedarsaniWZ2014, ansaripour2013robust}. In this case, the CCSP may have a preferred order in which the tasks are executed based on QoS guarantees it has promised to the users. However, if the servers are independent entities providing service for a fee, they may follow a different order of performing the tasks. This misalignment can cause the CCSP to violate the QoS guarantees it has promised, and hence, degrade system performance.

As another instance, employees of an organization may perform tasks (such as responding to emails in technical support) that are assigned to them in a different order than the one that is desired by the organization. Since the rate at which humans can respond to emails is limited, emails  pile up~\cite{kahneman,spira2011}. People generally do not respond to emails in the received order, but act on them based their own priorities~\cite{dabbish,gagne, isen2005,wainer} which may be based on factors that are both intrinsic (e.g., interest, curiosity, or information gaps) and extrinsic, (e.g., incentives provided by the organization). Thus, a similar problem as we consider arises in which the organization must incentivize employees to respond to tasks according to the order preferred by the organization.


In this paper, we model the problem as one of designing a contract through which the principal (the system operator) asks the agents about their private priorities and incentivizes them to shift their priorities in a way that is socially optimal. Specifically, since the agents incur a cost to change their priorities from their private ones, the principal needs to provide enough incentives so that rational agents will shift their priorities to align with those of the principal. For simplicity, in this work, we assume that only one agent is selected to execute all the tasks. The primary challenge in the design of the contract arises from the hidden nature of the priorities of the agents who are free to misreport them. Thus, a simple compensation scheme based on self-reported priority will not be sufficient as the agents can misreport the baseline or the private priorities. The private nature of the individual priorities at the agents causes adverse selection  and also leads to the effort put in by the agents to change from the initial preferred priority to the realized one to also be hidden. In other words, there is the problem of both \emph{hidden information} and \emph{hidden action} for the principal~\cite[Chapter~14C]{mas1995},~\cite[Chapter~14B]{mas1995}. Further, the principal can observe the priority realized by only the selected agent. Our goal is to design a contract which resolves these issues and incentivizes the agents to put in sufficient effort to realize a priority that optimizes the  social welfare. 

Our solution relies on formulating this problem as a two-stage contract design problem. In the first stage, the principal selects the agent to whom to allocate the tasks using the priorities self-reported by the agents. In the second stage, the principal compensates the selected agent using the priority he realizes in a way that leads to social optimality. We propose a VCG-based mechanism for the first stage in which the agents announce their private priorities to the principal and the principal selects the agent to assign the tasks. We show that the first stage limits  misreporting by the agents. In the second stage, we design a compensation scheme using the observable realized priority by the selected agent and the initially declared priorities. In this two-stage design, the agents bid (possibly falsified) priorities in the first stage and the selected agent optimizes the realized priority for performing the tasks in the second stage. The principal designs the auction in the first stage and the compensation in the second step.  
 
The model considered herein is inspired by~\cite{sharma}, which presents a queueing-theoretic study of the problem. However, unlike~\cite{sharma}, we do not consider the realized priority as a given and fixed function of the priority of the principal and interests of the agent, but as a design parameter for the agent to maximize his own utility. Although there is a vast literature on multi-agent task scheduling literature (see, e.g.,~\cite{ananth, bredin, el1994}), prior work does not consider either information asymmetry between the agents and the principal or the design of incentives. To the best of our knowledge, this paper is the first work to adopt a game-theoretic approach to analyzing priority misalignments between task senders and task receivers. In the mechanism design literature, VCG mechanisms have long been used for incentive design in the case of hidden information between the principal and the agents. In particular, VCG mechanisms are used to incentivize agents to reveal their true private information and to guarantee the efficient (socially optimal) outcome in dominant strategies~\cite[Chapter~23]{mas1995}, ~\cite[Chapter~5]{krishna}. However, a VCG-based mechanism is effective only for the first stage of our problem when we select the agent to perform the tasks and its interaction with the second stage which features compensation for the hidden effort put in by the agent to align his priority with the principal is \textit{a priori} unclear. 
 
Our main contribution is developing a game-theoretic approach to the problem of task allocation and priority realization when there is information asymmetry and possibility of misreporting private information by the agents. The problem features both hidden information and hidden action, and is significantly different than problems of pure adverse selection or pure moral hazard. We propose a VCG-based mechanism followed by an incentivization method for the problem. We show that under the proposed scheme, the agents act truthfully in reporting their preferred priorities in a dominant strategy manner. In addition, the principal can achieve the socially optimal outcome, as well guarantee individual rationality and incentive compatibility, through the proposed mechanism.

The rest  of the paper is organized as follows. Section~\ref{Sec2} presents the problem statement and some preliminaries. Section~\ref{sec3} proposes and analyzes our incentive  mechanism for the principal-agent queuing problem. Section~\ref{sec4} concludes by presenting  potential directions for future work.


\section{Problem Statement}
\label{Sec2}
Consider a group of $N+1$ decision makers. Decision maker $0$ is the principal, who is interested in performing a sequence of tasks with a particular priority. Decision makers $1, \ldots, N$ are agents with their own private priorities for performing the tasks. The principal must incentivize the agents to perform the tasks in the desired order.

\subsection{Model}
The principal seeks to delegate $M$ tasks, denoted by $k=1,\ldots, M$, to a group of self-interested agents,  denoted by $i=1,\ldots, N.$ 
%
 %
 All the decision makers have an associated priority with which they wish to execute the tasks. Let $X=[x_1, \ldots, x_M]$ denote the priority of the principal where $x_k$ is the priority for executing the $k$th task. Similarly, let $Y_i=[y_{i1}, \ldots, y_{iM}]$ denote the  priority for the $i$th agent, where $y_{ik}$ is the priority of the $i$th agent for fulfilling the $k$th task. The vector $X$ is public knowledge while the vector $Y_{i}$ is private knowledge to the $i$th decision maker.
 
\begin{remark}
 Note that in this paper, we assume that  priority vectors are metric data and not ordinal data.
 \end{remark}

The principal selects one agent and incentivizes him to execute the tasks in an order as close to $X$ as possible. Given the incentive, if the $i$th agent is selected to execute the tasks, let $Z_{i}$ be the {\em realized} priority of execution. Further, denote by $h(Y_i, Z_i)$ the effort cost for the agent to change his priority from $Y_{i}$ to $Z_{i}$. In other words, when the agent with priority $Y_{i}$ is selected, and he performs the tasks with priority $Z_{i}$, he incurs the cost $h(Y_i, Z_i)$.

If agent $i$ is selected, $Z_{i}$ is observable to the principal. In other words, the principal can observe the order in which the tasks were actually executed. Contrarily, if $i$ is not selected, neither $Y_{i}$ nor $Z_{i}$ is observable for that agent, since the agent is not assigned any task to realize the priority $Z_i$.


Note that since the principal does not have access to the  priorities $Y_{i}$'s of the agents, these variables are not contractible. In fact, if the principal inquires about the vectors $Y_i$'s, the agents can misreport them as $Y'_i$'s to try and exploit the incentive mechanism to gain more benefit.

 \subsection{Problem Formulation}
\begin{figure}[tbp]
\centering
\includegraphics[width=8.7cm, height=2.1cm]{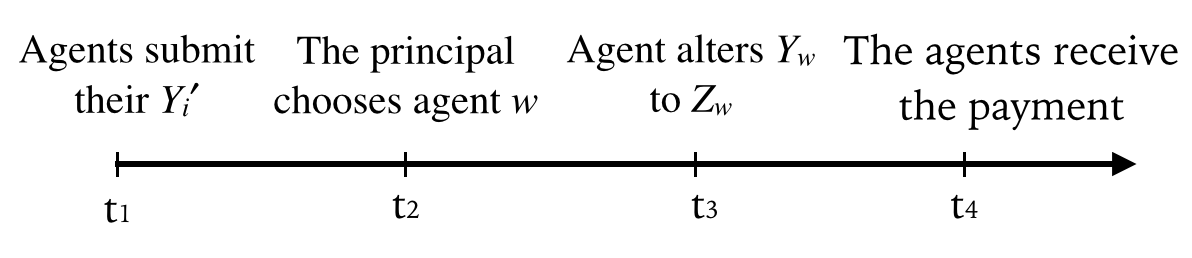}
\caption{Timeline of the interaction between the principal and the agent.}
\label{pic_2}
\end{figure}
Figure \ref{pic_2} demonstrates the timeline of the problem. The principal receives the (possibly false) reported priorities $\{Y'_{i}\}$ and chooses an agent $w$ based on an as yet undetermined mechanism. The principal then observes the realized priority $Z_{w}$ and pays every agent $i$ payment equal to $P_{i}(X,Y_{1}',\ldots,Y_{N}',Z_{w}).$ The mechanism to choose the agent $w$ as well as the payment are committed ex ante. Note that the priority modification to $Z_w$ by the agent enhances the performance of the organization and leads to profit $S(X,Z_w)$ for the principal. 

\begin{remark}
In this paper, we assume an indivisible array of tasks that must all be executed by one agent. Optimally allocating tasks to multiple agents is a significantly harder problem that is left for future work. 
\end{remark}
\begin{remark}
The choice of the agent to execute the tasks is  a challenging problem since the payment function is committed \textit{ex ante}.
\end{remark}

The utilities of the various decision makers are as follows. Suppose that agent $w$ is selected to execute the tasks. Then, the utility $U_{i}$ of the $i$th agent is given by
 \begin{equation}
U_i=
\begin{cases}
P_{i}(X,Y_{1}',\ldots,Y_{N}',Z_{w})-h(Y_w , Z_w )& i=w\\
P_{i}(X,Y_{1}',\ldots,Y_{N}',Z_{w})&i\neq w.
\end{cases}
 \end{equation}
The utility of the principal can be written as 
 \begin{equation}
 V=S(X, Z_w)-\sum_{i=1}^{N} P_{i}(X,Y_{1}',\ldots,Y_{N}',Z_{w}).
  \end{equation}

We are interested, in particular, in mechanisms that are socially optimal (or, in other words, efficient). An incentive mechanism is socially optimal if the decision makers choose to realize an outcome that maximizes the social welfare given by
\[
\Pi=V+\sum_{i=1}^{N} U_i.
 \]

 
The problem faced by each agent is to optimize the choice of reported priority $Y_i'$, and if chosen to perform the tasks the choice of realized priority $Z_i$, to maximize his utility (subject to the principal's choices). The problem faced by the principal is to choose the agent $w$ to execute the tasks and to design the payment $P_{i}(X,Y_{1}',\cdots,Y_{N}',Z_{w})$ to optimize the social welfare (subject to the choices of the agents). Thus, the problem we are interested in can be written as
\begin{equation}
\mathcal{P}_{1}:
\begin{cases}
\{w,P_{i}\}=\argmax \Pi\\
\textrm{subject to }Y_i'^{*}=\argmax ~U_i,~\forall i\neq w,\\
\{Y_w'^{*},~ Z_w^*\}=\argmax ~U_w\\
\text{additional constraints}
\end{cases}.
\label{prob2}
\end{equation}

We consider the following two additional constraints in~$\mathcal{P}_{1}$. 
\begin{enumerate}[(i)]
\item \textit{Individual Rationality (IR)}:  Individual rationality or participation constraint implies that under the incentive mechanism  
\[
V\geq 0,\qquad U_i\geq 0 \mbox{ for all } i~.
\]
Informally, the principal and the agents, acting rationally, prefer to participate in the proposed contract rather than opting out. This constraint limits the space of contracts by, e.g., precluding contracts based only on penalties.
\vspace{2mm}
\item \textit{Incentive Compatibility (IC)}:  A payment or a contract is incentive compatible if  the agents  submit their hidden information truthfully if asked.  Specifically, this constraint implies that the utility of an agent does not increase if they report $Y_{i}'\neq Y_{i}$; or, in other words, for any $i$
\begin{multline*}
   U_i(Y_1,\ldots,Y_{i-1},Y_i',Y_{i+1},\ldots,Y_{N})\\\leq U_i(Y_1,\ldots,Y_{i-1},Y_i,Y_{i+1},\ldots,Y_{N})~. 
\end{multline*}

\end{enumerate}

 \begin{asum}
For simplicity, we define the misalignment between two priority vectors $\Lambda$ and $\Gamma$ as a scalar function $m(\Lambda,\Gamma)$ of the two vectors. The function $m(\cdot,\cdot)$ can, for instance, be the norm of the difference of the two vectors. 
We assume that both the effort cost $h(Y_i , Z_i )$ and the profit $S(X , Z_i )$ are functions of the misalignment $m(Y_i , Z_i )$ and $m(X , Z_i )$ respectively. 
In addition, we define
\begin{align*}
\theta_i&=m(X,Y_i),\qquad \gamma_i=m(X,Z_{i}),\quad \theta'_i=m(X,Y_i'),
\end{align*}
where $\theta_i$ and $\gamma_i$ denote the initial priority misalignment and the realized priority misalignment between the agent $i$ and the principal respectively. Further,  $\theta'_i$ corresponds to priority misalignment declared by the agent initially.  In the sequel, we abuse the notation and denote the effort cost as $h(\theta_{i},\gamma_{i})$  and the profit of the principal as  $S(\gamma_i)$.
\end{asum}

We make the following two further assumptions. 
\begin{asum}
\label{asum44}
If agent $i$ is selected, the realized priority misalignment by the agent is always less than the initial  priority misalignment, i.e., $\gamma_i \leq\theta_i.$ In other words, the agent does not gain any benefit by increasing his priority misalignment with the principal. 
\end{asum}
Given this assumption, the principal can restrict the falsification by the agents in reporting their priorities through an appropriate payment function. Note that $\theta_i$ is unobservable to the principal even if agent $i$ is selected to execute the tasks. Thus, the principal must rely on $\theta'_i$ instead for the payment scheme. However, the principal may pay an agent only if $\gamma_i \leq\theta'_i$ to restrict the falsification by the agent. We assume that such a payment scheme is used and the following behavior is followed by the agents. 
\begin{asum}
The agent $i$, if selected, chooses $\gamma_i$ and $\theta'_i$ such that $\gamma_i \leq\theta'_i$.
\end{asum}

\section{Main Results}
\label{sec3}

We propose a two-step mechanism in which first an agent is selected to execute the tasks through an auction mechanism and then payments are made according to the reported and realized priorities. Note that the hidden nature of the preferred priorities and the effort cost creates the problem of  hidden information (adverse selection) in the first stage and then the problem of hidden action (moral hazard) in the second stage. The constraints of individual rationality and incentive compatibility significantly constrain the design of each of these steps. For instance, at the first stage, an auction which asks the agents to report their priorities and chooses the agents with the least reported priority misalignment will not be incentive compatible since it provides an opportunity for the agents to announce a priority close to that of the principal to be selected. 

Similarly, consider a payment scheme in which the agent $w$ that is selected to execute the tasks which depends  merely on the reported misalignment $\theta'_w$ regardless of the realized priority $Z_w$ (or equivalently $\gamma_{w}$) and ignores the effort cost. Given $\theta_w$, this payment limits the range of realized priority such that 
\[
h(\theta_w, \gamma_w)\leq p(\theta'_w)~.
\]
There is no \textit{a priori} guarantee that the resulting priority vector $Z_{w}$ will be socially optimal. On the other hand, a payment that is merely a function of the realized priority $\gamma_w$ and ignores the self-reported priorities may also be too restrictive. In particular,  individual rationality will once again constrain $Z_{w}$ (or equivalently $\gamma_{w}$) so that given $\theta_{w}$, \[h(\theta_w, \gamma_w)\leq p(\gamma_w)~.\]

Finally, we note that even if the payment depends on both $\theta_{w}'$ and $\gamma_{w}$ to account for the effort cost properly and satisfy individual rationality, the payment still needs to be carefully designed to ensure incentive compatibility. Thus, a payment function $p(\theta_w', \gamma_w)$ that depends on the level of effort cost that the agent claims that he incurred for priority modification provides the opportunity for the agents to behave strategically. For instance, under such a payment, a strategic agent may choose not to exert any effort and choose
\[
\gamma_w=\theta_w,\:\:\: \theta'_w=\argmax p(\theta_w', \gamma_w)~.
\]
Thus, this payment is not incentive compatible since although the strategic agent does not change his priority, he obtains a non-zero payment.

\subsection{Proposed Mechanism}
We now propose a two-step mechanism which attains the desired properties of individual rationality, incentive compatibility, and under further assumptions, social optimality. This mechanism first selects an agent to execute the tasks and then compensates him. 

Recall from Figure~\ref{pic_2}, that the timeline of the problem is as follows: \begin{enumerate}
\item The agents are asked to submit their preferred priority vectors $Y_{i}$'s (equivalently, the variables $\theta_{i}$'s). However, they can misreport the vectors as $Y'_{i}$'s (equivalently as $\theta_{i}'$'s). 
\item The principal chooses an agent as the winner of the auction. Assume that the agent with index $w$ is the winner. Agent $w$ is expected to execute the tasks in the next stage.
\item  Agent $w$ performs the task with a realized priority $Z_{w}$. In other words, it chooses the variable $\gamma_w$ and incurs the corresponding effort cost.  
\item The agent $w$ receives a payment. 
\end{enumerate}
\vspace{1mm}

We now present our proposed mechanism $\mathcal{M}$. 
\begin{enumerate}[(i)]
\item The principal chooses the agent $w$ to execute all the tasks such that $w=\argmin\{\theta_i'\}_{i=1}^{N}$. 
\item The payment to agent $w$ is chosen as a function of $\gamma_w$ and  the second lowest bid $\bar{\theta}=\min \{\theta'_{1}, \cdots, \theta'_{w-1},\theta'_{w+1}, \cdots,\theta'_{N} \}.$ Specifically, we consider a payment  $P_w(\bar{\theta},\gamma_w)$ to agent $w$ which satisfies two properties:
 \begin{equation}
 \label{eq:prop_payment}
\begin{split}
&\text{ $\forall\gamma_w$, if $\theta_w\geq \bar{\theta}$, we have $P_w(\bar{\theta},~\gamma_w)\leq h(\theta_w,~\gamma_w)~$}\\
& \text{$\exists \gamma_w$ s.t. if $\theta_w<\bar{\theta}$, we have $P_w( \bar{\theta},~\gamma_w)> h(\theta_w,~\gamma_w)$}
\end{split}.
 \end{equation}
 \item All other agents $i\neq w$ are not paid.
\end{enumerate}
\begin{remark}
Note the condition in~(\ref{eq:prop_payment}) is essential for inducing incentive compatibility. An example of a payment scheme which satisfies this condition for the  cost function  $h(\theta_w,~\gamma_w)=\theta_w-\gamma_w$ is of the form  $P_w(\bar{\theta},~\gamma_w)=\bar{\theta}-\gamma_w.$
\end{remark}

Under the proposed mechanism, the utilities of the agents  are given by
  \begin{equation}
  \label{agent}
 U_i=\begin{cases}P_i( \bar{\theta},~\gamma_i)- h(\theta_i,~\gamma_i)&  \text{ $i=w$}\\
 0&  \text{$i\neq w$}
 \end{cases},
 \end{equation}
while the utility of the principal and the social welfare can be written as
 \begin{equation}
V=S(\gamma_{w})-P_w(\bar{\theta},~\gamma_w), \quad \Pi= S(\gamma_{w})-h(\theta_w,~\gamma_w)~.
 \end{equation}

The following result shows that the proposed mechanism is incentive compatibile and individually rational. 
\begin{thm}
\label{th1}
Consider the problem $\mathcal{P}_{1}.$ The proposed mechanism $\mathcal{M}$  is incentive compatible, i.e., every agent $i$ reports $\theta_i'=\theta_i.$ Further, it satisfies the individual rationality constraint.
\end{thm}
\begin{proof}
 See Appendix.
 \end{proof}
 
\begin{remark}
\label{rem4}
The proposed mechanism resembles the celebrated VCG mechanism in the way it selects $w$ and in the structure of the proposed payment. However, beyond the fact that the payment depends on the additional parameter $\gamma_{w}$, note that the standard solution of offering a payment of the form $S(\gamma_w)-\Pi^{\star}$, where $\Pi^{\star}$ denotes the the value of the social welfare under the socially optimal outcome, will not result in the agent $w$ realizing the socially optimal outcome in our case. This payment violates~(\ref{eq:prop_payment}) and therefore violates incentive compatibility constraint.
 \end{remark}

Although social optimality is difficult to achieve for a general form of the effort cost, it can be achieved for the case of linear effort cost, i.e., when $h(\theta_i,~\gamma_i)=|\theta_i-\gamma_i|$.
\vspace{2mm}

\begin{thm}
\label{pro1}
Consider the problem $\mathcal{P}_{1}$ with the mechanism $\mathcal{M}$. If the effort cost is linear and the payment is chosen to be of the form $P_w=\bar{\theta}-\gamma_w,$ then $\mathcal{M}$ solves $\mathcal{P}_{1}$. Specifically, the mechanism
\begin{enumerate}[(i)]
\item guarantees truth-telling by the agents , i.e.,
$\theta'_i=\theta_i$, in (weakly) dominant strategy,
\item realizes the socially optimal outcome, and
\item is individually rational.
\end{enumerate}
\end{thm}
\begin{proof}
See Appendix.
\end{proof}

\subsection{Discussion}
The problem that we consider is challenging primarily because it exhibits both hidden information and hidden effort on the part of the agents without any recourse to verification. The combination of adverse selection and moral hazard creates a possibility of rich strategic behavior by the agents. We need to design both an auction and a compensation scheme. If the problem were of either auction design or compensation design alone, a rather standard mechanism can solve the problem. Specifically, for auction design when our focus is on a problem of adverse selection, we can achieve incentive compatibility, individual rationality, and social optimality by the VCG mechanism. Similarly, for compensation design to counteract pure moral hazard, social optimality and individual rationality can be realized through the standard contract of the form discussed in Remark~\ref{rem4}.

However, our problem features both adverse selection and moral hazard, and thus requires an auction followed by compensation for priority alignment. In this case, strategic agents can exploit the information asymmetry to degrade the efficiency of the outcome under either of the standard solutions for auction design or compensation for moral hazard alone. In other words, we can not achieve both incentive compatibility and social optimality. This result is similar in spirit to the so-called 
price of anarchy which captures the inefficiency in a system as a result of selfish/strategic behavior of the agents. 

The surprising result in Theorem~\ref{pro1} is that for a specific effort cost, we are able to realize all three properties of individual rationality, incentive compatibility, and social optimality even in this challenging setup. That this is possible was not \textit{a priori} obvious, and it would be interesting to identify further properties, such as budget balance, that may be achievable. Impossibility results such as~\cite{myerson,green} may seem contradictory to the goal of obtaining an efficient, individually rational, budget-balanced mechanism. It should be noted, however, that we are restricting attention to a specific class of valuation functions that are different from those studied in these results; hence, the existing impossibility results may not hold anymore.  Studying the existence of such behavior will be left as a future direction.

\section{Conclusion}
\label{sec4}
In this paper, we studied the problem of contract design between a system operator and a group of agents that each have a desired sequence of performing a collection of tasks. Since the priority orders for the agents is private information for them and these orders may not align with that of the principal, there is information asymmetry. The principal selects one of the agents to execute the tasks and wishes to realize the socially optimal outcome. The problem is to design a mechanism for selecting the agent to execute the tasks and to compensate him to minimize the misalignment of the realized priority with the one that is socially optimal. The problem features both moral hazard and adverse selection. 
We proposed a two-stage mechanism including a VCG-like mechanism for task allocation followed by a compensation mechanism. We showed that the mechanism is individually rational, incentive compatible, and for linear effort costs, socially optimal. 

Future work will consider the more general case where there are divisible tasks so that multiple agents need to be selected~\cite{sharma,ChatterjeeVV2015}. This problem adds task-scheduling to the mechanism design; in other words, the principal must solve a resource allocation problem followed by compensation design. Other directions include considering the possibility of designing a mechanism that is also budget balanced in addition to being individually rational, incentive compatible, and socially optimal.


\section*{Appendix}
\subsection{Proof of Theorem \ref{th1}}

Given an arbitrary agent $i$, its hidden priority $\theta_i$, and  the reported priority misalignment of the other players,  we need to show that  utility of agent $i$ is maximized by setting $\theta_i'=\theta_i$.  Note also that $\bar{\theta}$ denotes the lowest priority misalignment reported by the other agents.  If $\theta_i>\bar{\theta}$, then agent $i$ loses and receives utility $0$. If $\theta_i\leq\bar{\theta}$, then agent $i$ wins the tasks and receives utility $P_i( ~\bar{\theta},~\gamma_i)-h(\theta_i,~\gamma_i)$ for performance of task.

We consider two cases. First, if  $\theta_i> \bar{\theta}$, the highest utility that agent  $i$ can gain for any value of $\gamma_i$ is given by\[
\max\{0, P_i( ~\bar{\theta},~\gamma_i)-h(\theta_i,~\gamma_i) \}~.\]
According to~(\ref{eq:prop_payment}), we can obtain that
\[
\max\{0, P_i( ~\bar{\theta},~\gamma_i)-h(\theta_i,~\gamma_i) \}=0~.\]
 Thus, agent $i$ can achieve this utility by bidding his priority  truthfully (and losing the auction). Second, if $\theta_i\leq \bar{\theta}$, the highest utility that agent  $i$ can gain according to  our mechanism is \[\max\{0, P_i( ~\bar{\theta},~\gamma_i)-h(\theta_i,~\gamma_i) \}=P_i( ~\bar{\theta},~\gamma_i)-h(\theta_i,~\gamma_i),\]
and agent $i$ can achieve this utility by bidding his priority truthfully and winning the auction. Note that the utility of the agent for each case is always non-negative and therefore the mechanism satisfies individual rationality constraint.
 \vspace{2mm}
 
\subsection{Proof of Theorem \ref{pro1}}
 
 First notice that according to Assumption \ref{asum44}, we can write the effort cost as $h(\theta_i,~\gamma_i)=\theta_i-\gamma_i$.
\begin{enumerate}[(i)]

\item We first prove that the proposed mechanism design induces truth-telling as a  dominant strategy, i.e., it is incentive compatible. Similar to the proof of Theorem \ref{th1}, we consider two cases. First, if  $\theta_i> \bar{\theta}$, the highest utility that agent  $i$ can get is 
\[
\max\{0, P_i(\bar{\theta},~\gamma_i)-h(\theta_i,~\gamma_i) \}~.\] 

 Given $P_i(\bar{\theta},~\gamma_i)=\bar{\theta}-\gamma_i$ and $h(\theta_i,~\gamma_i)=\theta_i-\gamma_i$, the highest utility that agent  $i$ can get if $\theta_i> \bar{\theta}$ is 
\begin{equation}
\max\{0, \bar{\theta}-\gamma_i-(\theta_i-\gamma_i) \}=\{0, \bar{\theta}-\theta_i\}=0~,
\label{utilityl}
\end{equation}
 and agent $i$ gains this utility by bidding truthfully and losing the auction. Second, if $\theta_i\leq \bar{\theta}$, the highest utility that agent  $i$ can get is 
 \begin{equation}
 \max\{0, \bar{\theta}-\gamma_i-(\theta_i-\gamma_i) \}=\bar{\theta}-\theta_i~,
 \label{utilitys}
 \end{equation}
 and agent $i$ gains this utility by bidding his priority truthfully and winning the auction. Note that another approach to check incentive compatibility of $P_i$ is to see that $P_i$ satisfies \eqref{eq:prop_payment}.
 
\item Next, we show that the agents realizes $\gamma^{\star}_w$ through this payment. 
 The socially optimal outcome is obtained as
 \[
\gamma^{\star}_w=\argmax S(\gamma_w)-h(\theta_w,~\gamma_w)~.
 \]
 On the other hand, given \eqref{utilityl} and \eqref{utilitys}, the utilities of the agents are given by
  \begin{equation}
  \label{uagent}
 U_i(\theta_i,\bar{\theta})=\begin{cases}\bar{\theta}-\theta_i,& \text{ $i=w$}\\
 0, &  \text{ $i\neq w$}
 \end{cases},
 \end{equation}
  which does not depend on the value of $\gamma_i$. Thus, the agent is indifferent among his realized priorities and we conclude that  the agents realize the socially optimal outcome $\gamma_w^{\star}$.
  
 \item Note that the utility of the agent in \eqref{uagent} for each case is always non-negative and therefore the proposed mechanism satisfies individual rationality constraint. 
\end{enumerate}

\bibliographystyle{IEEEtran}
%

\bibliography{reference}
\balance

\end{document}